%% file: main.tex
\documentclass[copyright,creativecommons]{eptcs}
\providecommand{\event}{GandALF 2022}
\usepackage[utf8]{inputenc}
\usepackage[T1]{fontenc}
\usepackage{underscore}
\usepackage{tikz}
\usepackage{amsmath}
\usepackage{stmaryrd}
\usepackage{amssymb}
\usepackage{amsthm,amscd}
\usepackage{upref}
\usepackage[all]{xy}
\usepackage{hyperref}
\usepackage{algorithm}
\usepackage{algorithmic}
\usepackage{xifthen}
\usepackage{xparse}
\usepackage{forest}
\usepackage{xspace}
\usepackage{relsize}
\usepackage{wrapfig}
\usepackage{subfig}
\usepackage{graphicx}
\graphicspath{{figures/}}

\usepackage{cleveref}

\usetikzlibrary{arrows,shapes,automata,backgrounds,petri}
\usetikzlibrary{automata,positioning}

\theoremstyle{plain}
\newtheorem{thm}{Theorem}[section]

\newtheorem{lem}[thm]{Lemma}

\theoremstyle{definition}
\newtheorem{defn}[thm]{Definition}

\newtheorem{ex}[thm]{Example}

\usepackage{macros}

\author{Thomas Brihaye
\institute{University of Mons\\
Mons, Belgium}
\email{thomas.brihaye@umons.ac.be}
\and
Sophie Pinchinat
\institute{Université de Rennes, IRISA\\
Rennes, France}
\email{sophie.pinchinat@irisa.fr}
\and
Alexandre Terefenko
\institute{Université de Rennes, IRISA\\ Rennes, France}
\institute{University of Mons\\
Mons, Belgium}
\email{alexandre.terefenko@irisa.fr}
}

\def\titlerunning{Adversarial Formal Semantics of Attack Trees and Related Problems}
\def\authorrunning{T. Brihaye, S. Pinchinat \& A. Terefenko}

\title{Adversarial Formal Semantics of Attack Trees and Related Problems}

\date{\today}

\begin{document}

\maketitle

\begin{abstract}
Security is a subject of increasing attention in our actual society in order to protect critical resources from information disclosure, theft or damage. The informal model of attack trees introduced by Schneier, and widespread in the industry, is advocated in the 2008 NATO report to govern the evaluation of the threat in risk analysis. Attack-defense trees have since been the subject of many theoretical works addressing different formal approaches.

In 2017, M. Audinot et al. introduced a path semantics over a transition system for attack trees. Inspired by the latter, we propose a two-player  interpretation of the attack-tree formalism. To do so, we replace transition systems by concurrent game arenas and our associated semantics consist of strategies. We then show that the emptiness problem, known to be \NP-complete for the path semantics, is now \PSPACE-complete. Additionally, we show that the membership problem is \coNP-complete for our two-player interpretation while it collapses to \PTIME in the path semantics.
\end{abstract}

\section{Introduction}
\label{sec:introduction}
\input{sec-introduction}

\section{Introductory example}
\label{sec:example}
\input{sec-example}

\section{Preliminary notions}
\label{sec:preliminaries}
\input{sec-preliminaries}

\section{Attack trees and their semantics}
\label{sec:ATsemantics}

In this section, we start with the formal definition of attack tree used in this paper. Then we develop two semantics for attack trees: the path semantics and the strategy semantics.

\subsection{Syntax of attack trees}
\label{sec:AT}
\input{sec-AT}

\subsection{Path semantics for attack trees}
\label{sec:singleplayersem}
\input{sec-singleplayersem}

\subsection{Strategy semantics for attack trees}
\label{sec:multiplayersem}
\input{sec-multiplayersem}

\section{Decision Problems over attack trees}
\label{sec:decisionproblems}
\input{sec-decisionproblems}

\section{Future work}
\label{sec:future}
\input{sec-future}

\bibliography{ref} 


\end{document}

%% file: sec-introduction.tex

Security is a subject of increasing attention in our actual society in
order to protect critical resources from information disclosure, theft
or damage. The informal model of attack trees was first introduced by
Schneier \cite{schneier1999attack} to schematically model possible
threats one could execute against an information system. Attack trees
have then been widespread in the industry and are advocated in the 2008
NATO report to govern the evaluation of the threat in risk
analysis. The attack tree model is also a subject of increasing attention in the
community of formal methods with a lot of different formal
approaches \cite
{mauw2005foundations,kordy2010foundations,jurgenson2008computing,
jhawar2015attack,horne2017semantics,
pinchinat2019attack,audinot2017my} (see the survey \cite{widel2019beyond}).

The first formal model of attack trees introduced
in \cite{schneier1999attack} aimed at describing a possible attack
over a system by refining the main attack goal into sub-goals using
either an operator $OR$ or an operator $AND$ to coordinate those
refinements. The analysis conducted over those trees is "static" in
the sense that the attacked system does not evolve during the
attack. As such, there is no concept of a goal happening before or
after another. In \cite{jhawar2015attack}, the authors introduce a first
formal semantics that can be qualified as "dynamic" by allowing a new
operator, operator $SAND$ (for sequential $AND$), to specify that
sub-goals must be attained in a given order. Considering that the $SAND$
operator is now commonly accepted, the authors of \cite{audinot2017my} propose a
path semantics for attack trees over a transition system.

In this paper, our goal is to present a new semantics for attack trees
in order to be able to model more realistic scenarios: we want our
attacker to be able to adapt her actions according to the behavior of
the environment -- typically, a defender who tries to protect the
system. This setting naturally yields a two-player semantics. Our
approach is inspired by \cite{audinot2017my} where we generalize the
path semantics to a game-theoretic framework, yielding a strategy
semantics for attack trees, without changing their syntax. 

While the path semantics from \cite{audinot2017my} is compositional in
a natural way, it turns out that the strategy semantics does not have
such a nice property. Indeed, although composition of strategies is
possible (see for example \cite{paul2015automata}), there is no
immediate solution to compose two strategies in order to get a
strategy that achieves the disjunction of what the formers
achieve. This situation makes it difficult to design a compositional strategy semantics for attack trees. We therefore develop a non-trivial strategy semantics
that is not compositionally obtained \emph{per se}, but that makes use of the
former compositional path semantics.

To our knowledge, our proposal is the first game-theoretic semantics
for attack trees, and should not be mixed-up with the multi-player
setting induced by the so-called classic model of \emph{attack-defense
trees} in the literature (see \cite{kordy2014attack}): attack-defense
trees are attack trees equipped with a new operator to express that
some defender could use a countermeasure to prevent attacker from
achieving her goal. Although well-understood in a "static" framework,
we are not aware of any formal semantics of attack-defense trees for
the ``dynamic'' one, \ie over transition systems. This missing piece
of work makes it difficult to conduct a comparison with our
contribution, but, from the fact that the defender in our setting is fully
formalized, as an opponent in the game arena, while the defender in
attack-defense trees takes the form of an abstract entity, it seems that those two formalisms are not expressing the same kind of problems.

Our contribution in this paper is twofold.

We first develop a clean mathematical setting to obtain a formal
strategy semantics for attack trees. For pedagogical reasons, we
choose to consider a simplified version of the attack trees
of \cite{audinot2017my} where atomic goals (at the leaves of the
trees) are reachability goals with no preconditions. However, at
the price of tedious definitions, the strategy semantics we propose
can be adapted to atomic goals with preconditions. Regarding the
design of this semantics, we heavily rely on the older one
of \cite{audinot2017my} based on paths, and we justify our approach by
providing evidence that a compositional strategy semantics is
hopeless.

Second, we exploit the attack tree semantics to address and study the
complexity of two decision problems: the \emph{non-emptiness problem}
and the \emph{membership problem}. The former consists in determining
if the semantics of an input tree is non-empty, while the latter
consists in determining if an input element belongs to the semantics of
an input tree. Our results are summarized in \Cref{table:results},
where we distinguish between the path semantics and the strategy
semantics. 
\begin{table}[h]
\begin{center}
\caption{Complexity results\label{table:results}}
\begin{tabular}{|c|c|c|}
\cline{2-3}
\multicolumn{1}{c|}{} & Paths semantics & Strategy semantics\\
 \hline
 Non-Emptiness Problem & \NP-complete & \PSPACE-complete\\
 \hline
 Membership Problem & \PTIME & \coNP-complete\\
 \hline
\end{tabular}
\end{center}
\end{table}

Importantly, both decision problems have a practical counterpart. The
non-emptiness of the path semantics of a tree reflects a situation
where there exists a favorable scenario for attacker to perform her
attack, while the non-emptiness of the strategy semantics reflects an
intrinsic vulnerability of an information system. Regarding the
membership problem for the path semantics, we are interested in
knowing whether a log file of some information system execution makes
evidence of an attack, while for the strategy semantics, we wonder if
an attack policy is successful.

The paper is organized as follows. We start in \Cref{sec:example}
with an introductory example explaining informally the difference
between path semantics and strategy semantics. After some background work in \Cref{sec:preliminaries}, we 
introduce in \Cref{sec:ATsemantics} the formal model of attack trees
and define the path semantics inspired by \cite{audinot2017my} as well
as our new strategy semantics. We then study the complexity of the
Non-Emptiness and the Membership problems
in \Cref{sec:decisionproblems}.

%% file: sec-example.tex
Consider a thief (the attacker) who wants to steal some document inside a safe of a building without being seen. The building is composed of two rooms. The first room has two entrance doors (called door 1 and door 2) from the street. There is a guard keeping the entrance doors, but he can only control the bypassing of one of the two entrance doors at a time. The first room has also another door that leads to the second room. This door is locked but the key to unlock it is in the first room.  There is also a camera in the first room monitoring the door to the second room. The second room contains a safe with the document that the thief wants to steal. Therefore, in order for the thief to attain his goal, he needs to enter the first room (by either door that is not currently controlled by the guard), then deactivate the camera and collect the key (in whichever order he wants) and finally unlock and go through the door leading to the second room. The thief can be seen by the guard if they appen to be in front of the same door or by the camera if activated when he is in front of the door to room 2. Figure \ref{Fig:build} gives a picture of the situation where the thief is still outside of the building and the guard controls the second door.

\begin{figure}[h]
   \centering
   \subfloat[Building plan]{
     \def\svgwidth{100pt}
     \input{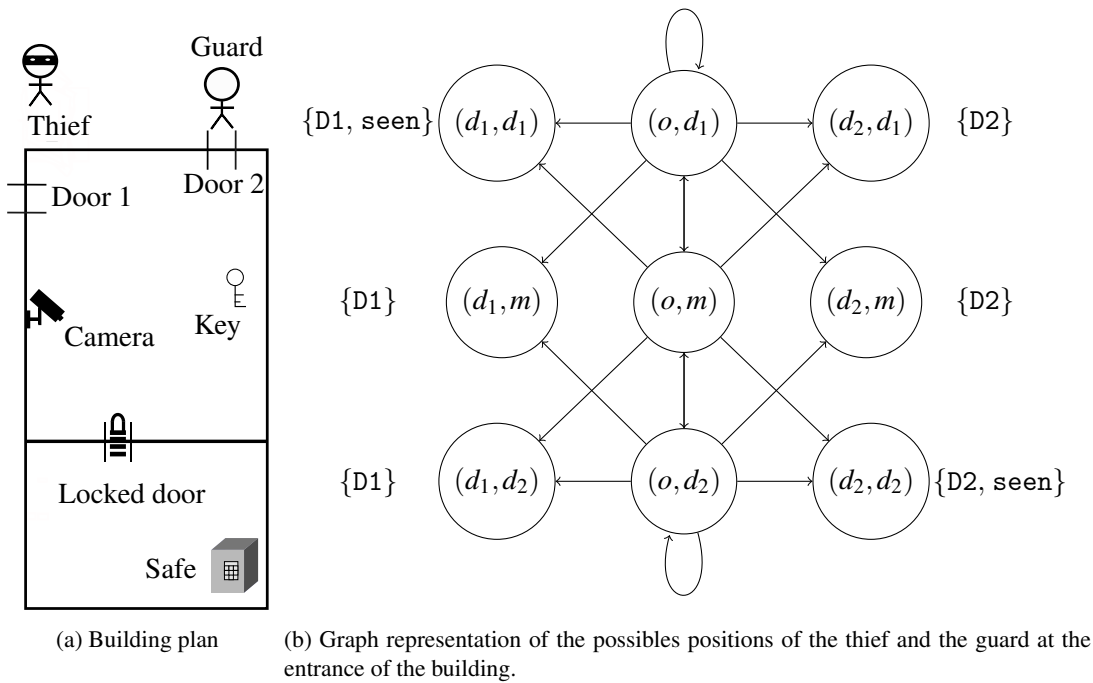}
     \label{Fig:build}
   }
   \subfloat[Graph representation of the possibles positions of the thief and the guard at the entrance of the building.]{
\begin{tikzpicture}[scale=0.5]

   \node[state] (q)   {$(o,m)$};
   \node[state] (q0)  [above=of q] {$(o,d_1)$}; 
   \node[state] (q1) [below=of q] {$(o,d_2)$};
   \node[state] (q2) [right=of q0] {$(d_2,d_1)$};
   \node[state] (q3) [right =of q1] {$(d_2,d_2)$};
   \node[state] (q4) [right=of q] {$(d_2,m)$};
   \node[state] (q2') [left=of q0] {$(d_1,d_1)$};
   \node[state] (q3') [left =of q1] {$(d_1,d_2)$};
   \node[state] (q4') [left=of q] {$(d_1,m)$};
   \node (blu) at (8,4.8) {$\{\propdtwo\}$}; 
   \node (blu) at (8,0) {$\{\propdtwo\}$};   
   \node (blu) at (8.4,-4.8) {$\{\propdtwo,$ \propseen $\}$};
   \node (blu) at (-8.4,-4.8) {$\{\propdone\}$}; 
   \node (blu) at (-8.4,0) {$\{\propdone\}$};   
   \node (blu) at (-8.4,4.8) {$\{\propdone,$ \propseen $\}$};  
   
   \path[->] 
    (q) edge node {} (q1)
    edge node {} (q0)
    edge node {} (q2)
    edge node {} (q3)
    edge node {} (q2')
    edge node {} (q3')
    (q0) edge node {} (q2)
    edge node {} (q4)
    edge node {} (q2')
    edge node {} (q4')
    edge node {} (q)
    edge [loop above] node {} ()
    (q1) edge node {} (q3)
    edge node {} (q4)
    edge node {} (q3')
    edge node {} (q4')
    edge node {} (q)
	edge [loop below] node {} ()
         ;
\end{tikzpicture}
\label{ga}
   }
   \caption{Introductory example building}
   \label{blublublu}
\end{figure}

In security, it is common to use an attack tree to model the goal of the attacker. An attack tree is a tree where each node describes a goal and the children of a node describe a refinement into sub-goals of the parent goal. To model those refinements, it is common to use three operators:
\begin{itemize}
\item $OR$ operator means that at least one sub-goal needs be achieved to have the goal accomplished,
\item $SAND$ operator, read "sequential and", means that the sub-goals need be achieved in the left-to-right order to have the goal accomplished, 
\item $AND$ operator means that the sub-goals need be achieved (in whatever order) to have the goal accomplished.
\end{itemize}

In Figure \ref{Fig:att}, we describe the goal of the thief by means of an attack tree. To distinguish the different types of nodes, we draw a curved line below an $AND$ operator, a curved arrow below a $SAND$ operator and nothing below an $OR$ operator.

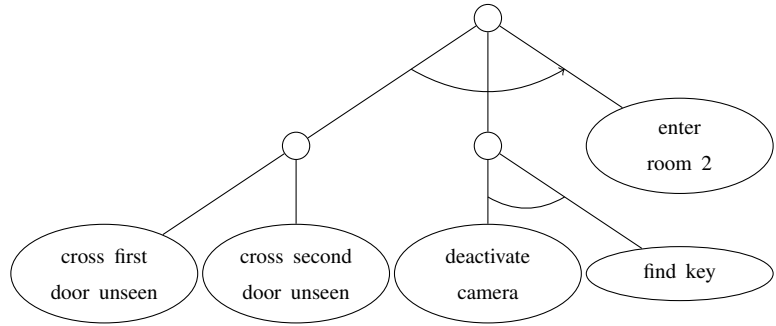
\begin{figure}
\begin{center}
\begin{tikzpicture}[scale=0.85,every text node part/.style={align=center}, node distance=3cm]
\node [draw, ellipse] (p) at (0,0) {};
\node [draw, ellipse] (q) at (-3,-2) {};
\node [draw, ellipse, text width = 1.5cm] (r) at (3,-2) {\scriptsize {enter room 2}};
\node [draw, ellipse] (truc) at (0,-2) {};
\draw (q) -- (p) -- (r);
\draw[->] (-1.2,-0.8) to[bend right] (1.2,-0.8);
\node[draw, ellipse, text width = 1.5cm] (s) at (-6,-4) {\scriptsize {cross first door unseen}};
\node[draw, ellipse, text width = 1.5cm] (t) at (-3,-4) {\scriptsize {cross second door unseen}};
\node[draw, ellipse, text width = 1.5cm] (u) at (0,-4) {\scriptsize {deactivate camera}};
\draw (s) -- (q) -- (t);
\draw (truc)--(p);
\draw[-] (0, -2.8) to[bend right] (1.2,-2.8);
\node[draw, ellipse, text width = 1.5cm] (v) at (3,-4) {\scriptsize {find key}};
\draw (u)--(truc) -- (v);
\end{tikzpicture}
\end{center}


\caption{An attack tree to model the goal of our thief.}
\label{Fig:att}
\end{figure}

We focus our attention on the very first part of our problem, that is, the sub-goal of the thief to enter the building. We start by fixing a set of proposition $\Prop=\{\propdone, \propdtwo, $\propseen$\}$ where $\propdone$ holds when the thief is in front of the first door, $\propdtwo$ holds when the thief is in front of the second door and \propseen\space holds when the guard sees the thief because they meet in front of the same door. The goal "crossing the first door unseen" of the thief can be modelled by the formula $\propdone \land \lnot\propseen$ and the goal "crossing the second door unseen" is modelled by the formula $\propdtwo \land \lnot \propseen$. If we assume that the guard can switch door whenever he wants but leaves the two doors unguarded for a brief amount of time during his motion, we can model the situation using the graph of Figure \ref{ga} where each state consists of a pair in the set $\{o, d_1, d_2\} \times \{m, d_1, d_2\}$. For a pair $(a,b)$, element $a$ determines the position of the thief ($o$ means that he is still outside the building, while $d_i$ indicates he is at door $i$) and element $b$ determines the position of the guard ($m$ means that he is currently in motion between the two doors, leaving them both unguarded). We write next to each state which propositions hold in it.

In a path semantics, a successful attack for goal $\propdone \land \lnot \propseen$ consists of a sequence of states (\ie a word) such that the valuation of the last state of this sequence satisfies formula $\propdone \land \lnot \propseen$. In particular, the sequence $(o,m), (d_1,d_2)$ is a successful attack. Now, if we want to consider a successful attack for the attack tree consisting of the "OR" operator of objective $\propdone \land \lnot \propseen$ and $\propdtwo \land \lnot \propseen$, it is enough to consider the union of the set of all attacks for objective $\propdone \land \lnot \propseen$ with the set of all attacks for objective $\propdtwo \land \lnot \propseen$.

However, in the strategy semantics we introduce in this paper, we consider that an attack is successful if the attacker has a strategy that grants him to reach the states he needs to attain, independently of the environment. In our example, we can see that the thief has no strategy starting at position $(o,m)$ to achieve goal $\propdone \land \lnot \propseen$. Indeed, if the first move of the guard consists on going to door 1 and he then does not move any more, there is no way for the thief to cross the first door while unseen. Similarly, there is no strategy starting at position $(o,m)$ to achieve goal $\propdtwo \land \lnot \propseen$. However, if we consider the "OR" operator of the two goals $\propdone \land \lnot \propseen$ and $\propdtwo \land \lnot \propseen$, the strategy of the thief consisting in waiting for the guard to go to one of the two doors and then going to door 1 if the guard is at door 2 and vice versa is a successful strategy. Later, we will use similar a similar example to show that a compositional definition for a strategy semantics cannot be achieved.

%
%

%% file: figures/fig-running_example.pdf_tex
\begingroup%
  \makeatletter%
  \providecommand\color[2][]{%
    \errmessage{(Inkscape) Color is used for the text in Inkscape, but the package 'color.sty' is not loaded}%
    \renewcommand\color[2][]{}%
  }%
  \providecommand\transparent[1]{%
    \errmessage{(Inkscape) Transparency is used (non-zero) for the text in Inkscape, but the package 'transparent.sty' is not loaded}%
    \renewcommand\transparent[1]{}%
  }%
  \providecommand\rotatebox[2]{#2}%
  \ifx\svgwidth\undefined%
    \setlength{\unitlength}{181.15109253bp}%
    \ifx\svgscale\undefined%
      \relax%
    \else%
      \setlength{\unitlength}{\unitlength * \real{\svgscale}}%
    \fi%
  \else%
    \setlength{\unitlength}{\svgwidth}%
  \fi%
  \global\let\svgwidth\undefined%
  \global\let\svgscale\undefined%
  \makeatother%
  \begin{picture}(1,2.20677492)%
  \put(0,0){\includegraphics[width=\unitlength,page=2]{fig-running\string_example.pdf}}%
    \put(0,0){\includegraphics[width=\unitlength,page=1]{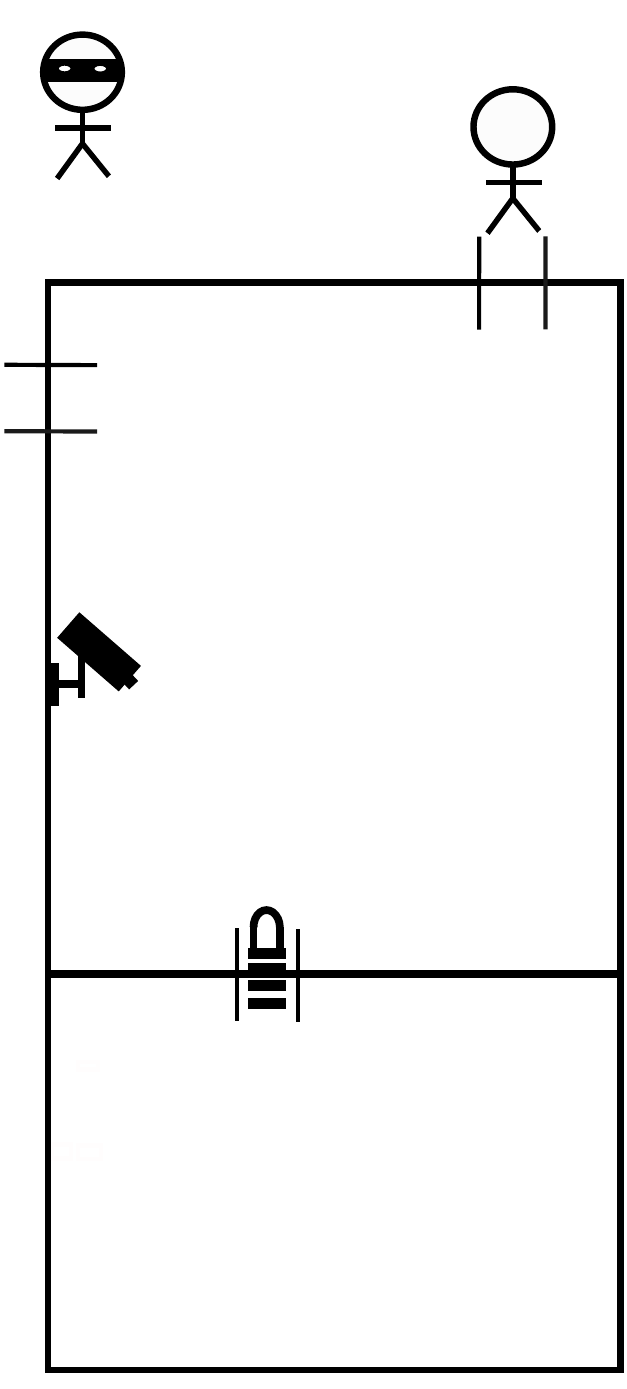}}%
    \put(0.71400975,1.06851753){\color[rgb]{0,0,0}\makebox(0,0)[lb]{\smash{Key}}}%
    \put(0.20096515,0.42776042){\color[rgb]{0,0,0}\makebox(0,0)[lb]{\smash{Locked door}}}%
    \put(0.17465621,1.55153761){\color[rgb]{0,0,0}\makebox(0,0)[lb]{\smash{Door 1}}}%
    \put(0.70143605,2.1205403){\color[rgb]{0,0,0}\makebox(0,0)[lb]{\smash{Guard}}}%
    \put(0.22397778,1.02294026){\color[rgb]{0,0,0}\makebox(0,0)[lb]{\smash{Camera}}}%
    \put(0.08397778,1.82){\color[rgb]{0,0,0}\makebox(0,0)[lb]{\smash{Thief}}}%
    \put(0.52924345,0.14573887){\color[rgb]{0,0,0}\makebox(0,0)[lb]{\smash{Safe}}}%
    \put(0.67400975,1.6){\color[rgb]{0,0,0}\makebox(0,0)[lb]{\smash{Door 2}}}%
  \end{picture}%
\endgroup%

%% file: sec-preliminaries.tex
\paragraph{Formal languages}
Given an alphabet (\ie a finite set of symbols) $\alphabet$, notation $\alphabet^*$
represents the set of finite words (\ie sequences of symbols) over
the alphabet $\alphabet$, with typical element $\word=\letter_1 \ldots \letter_n \in \alphabet^*$; the empty word is written $\eword$. On
the other hand, the set of infinite words is denoted by
$\alphabet^\omega$. A subset $L \subseteq \alphabet^*$ is a language. 
%
Given a word $\word=\letter_1 \ldots \letter_n \in \alphabet^*$, its set of
prefixes is the language $\Prefixes(\word)=\{\letter_1 ... \letter_i
| i\leq n\} \cup \{\eword\}$, and we write 
$\word' \prefix \word$ whenever $\word' \in \Prefixes(\word)$.
%
A language $\lang$ is prefix-closed if $\word \in \lang$ implies $\word'\in \lang$ for every $\word' \prefix \word$.
The concatenation of a word $\word=\letter_1 \ldots \letter_n$ with another word $\word'=\letter_1' \ldots \letter_m'$ is the word $\word\word'=\letter_1 \ldots \letter_n\letter_1' \ldots \letter_m'$. The concatenation of a word with a language is defined in the usual
way: for $\word \in \alphabet^*$ and $\lang \subseteq \alphabet^*$,
we let $\word\lang:=\{\word\word'| \word' \in \lang\}$. 

\paragraph{Game theory and game arena}
To formalize a strategy semantics for attack trees, we need standard two-player, zero-sum, perfect information games that we recall here.

A \emph{game arena} is a finite graph on which two players play a game of unbounded duration. We choose to consider concurrent games, meaning that, at each round, each player makes an action. To define it properly, we consider two players called Player $1$ and Player $2$ and a finite set of propositions $\Prop$. 

\begin{defn} 
A \emph{two-player game arena} is a tuple $\arenaa$ where:
\begin{itemize}
\item $\Pos=\{\pos, ...\}$ is a finite set of positions.
\item $\Act=\Act_1 \times \Act_2$ is a finite set of actions, as a product of the sets of actions of each player.
\item $\arenaf: \Pos  \times \Act \to \Pos$ is a transition function,
\item $\arenav:\Pos \to \mathcal{P}(\Prop)$ is a valuation function.
\end{itemize}
\end{defn}

In our introductory example of \Cref{sec:example}, if we consider the set of actions for the thief: $\{$wait, go-to-Door-1, go-to-Door-2$\}$ and the following set of actions for the defender: $\{$stay-at-current-door, leave-current-door, go-to-Door-1, go-to-Door-2$\}$, then it is easy to see that the graph drawn in Figure \ref{ga} forms a game arena.

For the rest of this paper, we fix a game arena $\arenaa$. 

%
%

For a game position $\pos \in \Pos$, we define
$\nextpos{\pos}=\{\pos'\in \Pos| $ $\arenaf(\pos, \act)=\pos'$ for
some $\act \in \Act \}$ the set of all positions reachable from $\pos$
in one step. For convenience, we assume that each player $\player$ can play every action $\act$ in $\Act_\player$ at each position of the game arena, so that for each
position $\pos \in \Pos$, we have $\nextpos{\pos} \neq \emptyset$. A play $\play$ is an infinite sequence of positions of the form
$\pos_0 \pos_1 \pos_2 .... \in \Pos^\omega$ such that for each
$i \in \mathbb{N}$, there is $\act \in \Act$ such that
$ \arenaf(\pos_i,\act)= \pos_{i+1}$. For $i \in \mathbb{N}$, we let
$\play_i=\pos_i$ be the $i^{th}$ position of play $\play$. The set
of all plays is denoted by $\Plays(\arena)$. Each non-empty prefix
$\hist$ of a play is called a \emph{history} and the set of all
histories is denoted by $\Hist(\arena)$. For a history
$\hist \in \Hist(\arena)$, we define $last(\hist)$ as the last
position of $\hist$. For $\pos \in \Pos$, we also use the notations
$\Plays(\arena, \pos)$ and $\Hist(\arena, \pos)$ to denote the set of
all plays starting from $\pos$ (\ie $\play_0=\pos$) and the set of all histories starting
from $\pos$, respectively.

Winning plays for Player $1$ are obtained from a distinguished subset $\win_1 \subseteq \Plays(\arena)$. As we consider zero-sum games, all plays in $\win_1\backslash \Plays(\arena)$ are winning for Player $2$.

Classically, we introduce the notion of \emph{strategy}, as a map prescribing how a player plays depending on the current history: a \emph{strategy} $\strat^\player$ for the Player $\player$ is a map $\strat^\player:\Hist(\arena) \to \Act_i$. The set of all strategies for the Player $\player$ is denoted as $Strat^\player$.

A history $\hist=\pos_0\pos_1...\pos_m$ is \emph{consistent} with a strategy $\strat^\player$ if for each $1 \leq i \leq m$, there exists $\act=(\act_1, \act_2) \in \Act$ such that we have $\arenaf(\pos_i,\act)=\pos_{i+1}$ and $\strat^\player(\pos_0\pos_1...\pos_i)=\act_\player$. We say that a play $\play$ is consistent with a strategy $\strat^\player$ if all prefixes of $\play$ are histories consistent with $\strat^\player$. The set of all plays consistent with $\strat^\player$ is denoted by $\out(\strat^\player)$. From a game position $\pos$ we say that a strategy is \emph{winning} if all outcomes starting at $\pos$ are winning. In other words, for a strategy of say Player $1$, a strategy $\strat$ is winning if $\out(\strat) \cap \Plays(\arena, \pos) \subseteq \win_1$.

A classic kind of concurrent games are the \emph{reachability games}. In such games, a player wants to reach some positions while the other player tries to prevent it from happening. These games are clearly zero-sum. More formally, we say that a game is a reachability game for Player $1$ if there exists $\reachwin \subseteq \Pos$ such that $\win_1=\{\play \in Plays(\arena)| \play_i \in \reachwin$ for some $i \in \mathbb{N}\}$.

In a game arena, we says that a position $\pos$ satisfies the formula $\formula$ if its valuation $\arenav(\pod)$ satisfies the formula in classic propositional logic, denoted by $\pos \models \formula$. Note that a Boolean formula $\formula$ over $\Prop$ describes the reachability game $(\arena, \reachwin)$ where $\reachwin=\{\pos \in \Pos|\pos \models \formula\}$.

%% file: sec-AT.tex
To formalize attack trees, we start by fixing a set of propositions $\Prop$.

\begin{defn}
An \emph{attack tree} $\att$ over $\Prop$ is:
\begin{itemize}
\item either a leaf composed of a unique Boolean formula $\formula$ over $\Prop$,
\item or an expression $OP(\att_1,...,\att_n)$ where $OP$ ranges over $OR$, $AND$ and $SAND$ and $\att_1, ...,\att_n$ are attack trees.
\end{itemize}
\end{defn}

We define the size of an attack tree $\att$, noted $\sizeatt{\att}$, by the number of its nodes.

\begin{ex}
\label{ex:sgattacktree}
We will formalise the example introduced in Section \ref{sec:example}. To represent the situation, we use the following set of propositions:  $\Prop=\{\propdone,\propdtwo, \propseen, \texttt{C},\texttt{K},\texttt{R2}\}$. We have that $\propdone$ holds when the thief has crossed the first door, $\propdtwo$ holds when the thief has crossed the second door, $\propseen$ holds when the guard sees the thief, $\texttt{C}$ holds when the camera is on, $\texttt{K}$ holds when the thief has the key and finally $\texttt{R2}$ holds when the thief is in the second room.

We can now  propose a formal definition for the attack tree in Figure \ref{Fig:att}: $\att= \break SAND(OR(\propdone \land \lnot \propseen, \propdtwo \land \lnot \propseen), AND(\texttt{C},\texttt{ K}), \texttt{R2})$ to model the objective of the attacker. The graph representation of $\att$ is given by Figure \ref{Fig:sgattacktree}.

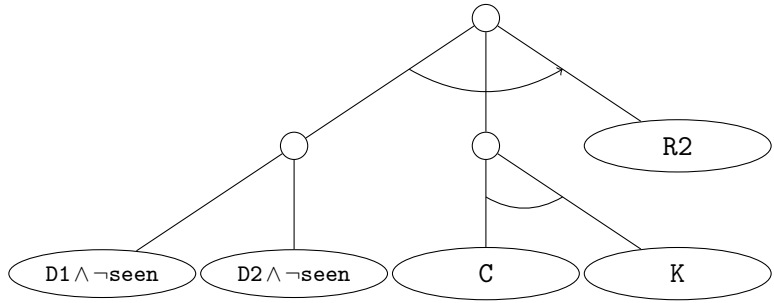
\begin{figure}
\begin{center}
\begin{center}
\begin{tikzpicture}[scale=0.85,every text node part/.style={align=center}]
\node [draw, ellipse] (p) at (0,0) {};
\node [draw, ellipse] (q) at (-3,-2) {};
\node [draw, ellipse, text width = 1.5cm] (r) at (3,-2) { $\texttt{R2}$};
\node [draw, ellipse] (truc) at (0,-2) {};
\draw (q) -- (p) -- (r);
\draw[->] (-1.2,-0.8) to[bend right] (1.2,-0.8);
\node[draw, ellipse, text width = 1.5cm] (s) at (-6,-4) {\footnotesize{$\propdone \land \lnot \propseen$}};
\node[draw, ellipse, text width = 1.5cm] (t) at (-3,-4) { \footnotesize{$\propdtwo \land \lnot \propseen$}};
\node[draw, ellipse, text width = 1.5cm] (u) at (0,-4) { $\texttt{C}$};
\draw (s) -- (q) -- (t);
\draw (truc)--(p);
\draw[-] (0, -2.8) to[bend right] (1.2,-2.8);
\node[draw, ellipse, text width = 1.5cm] (v) at (3,-4) { $\texttt{K}$};
\draw (u)--(truc) -- (v);
\end{tikzpicture}
\end{center}
\end{center}
\caption{Formal version of the attack tree in Figure \ref{Fig:att}.}
\label{Fig:sgattacktree}
\end{figure}

\end{ex}

Let us notice that the attack trees used in \cite{audinot2017my} are in fact slightly different: the leaves of the attack trees of that paper are of the form $<\formula_1, \formula_2>$ with $\formula_1$ and $\formula_2$ two Boolean formulas. Formula $\formula_1$ describes a precondition for the objective to begin with and formula $\formula_2$ describes the postcondition for the objective to be granted. However, in this paper, we never consider preconditions. So a leaf $\formula$ of our attack trees can be seen as a leaf $<true, \formula>$ of attack trees introduced in \cite{audinot2017my}. Although the semantics defined in this paper could also be defined for attack trees with preconditions, not considering them makes the setting more pedagogical.

The first semantics for attack trees we introduce is a path semantics inspired by \cite{audinot2017my}. Informally, in the path semantics, we consider all sequences of events that lead to a successful attack. The idea is to determine which scenarios are favourable for the attacker. One could also say that an attack can occur if the attacker is lucky.

The second semantics is our main contribution, and is named the \emph{strategy semantics} for attack trees. In this approach, the attacker should not rely on an opportunity offered by the environment but should be able to find the right sequence of actions whatever the environment does. Otherwise said, an attack is not a favourable scenario anymore but a winning strategy for attacker in some two-player game arena.

%% file: sec-singleplayersem.tex

To give a path semantics over our trees, we first need to fix a transition system to model which actions/sequences of actions can be executed by the attacker in the system. A transition system is composed of a finite set of states together with a transition relation between pairs of states. We decided not to label every transition with an action as we only consider perfect information here. We also provide a valuation function to our transition system, informing which propositions of $\Prop$ holds in a state of the transition system.

\begin{defn}
\label{def:trans}
A \emph{transition system} over $\Prop$ is a triplet $\transall$ where:
\begin{itemize}
\item $\transS$ is a finite set of states,
\item $\transf \subseteq \transS\times \transS$ is a relation of transitions,
\item $\transv:S \to \mathcal{P}(\Prop)$ is a valuation function. 
\end{itemize}
\end{defn}

 The size of $\trans$, noted $\sizetrans{\trans}$, is defined by its number of states.

We can see from  \Cref{def:trans} that a transition system is a notion close to a game arena. Indeed, it is easy to associate a transition system with a game arena $\arenaa$, by merging the two players into a single one in the following way: $\trans_\arena =(\Pos, \{(\pos, \pos') \in \Pos \times \Pos |$ there exists $\act \in \Act$ such that $\arenaf(\pos, \act)=\pos'\},\arenav)$. Later, we denote $\trans_\arena$ by $\arena$ when it is clear from the context.

For the rest of this section, we fix a transition system $\transall$  over $\Prop$. A \emph{path} in a transition system is a finite non-empty sequence of states $\path=\transs_0\transs_1...\transs_n$ such that, for each $0\leq i <n$,  $(\transs_i, \transs_{i+1}) \in \transf$. The size of a path is its number of states. We denote the set of all paths in $\trans$ by $\Path_\trans$. 

In order to define the path semantics we need to introduce operators over paths.

\begin{defn}
Let $\path=\transs_0\transs_1...\transs_n$ and $\path'=\transs'_0\transs'_1...\transs'_m$ be two paths in $\trans$ with $n,m\geq 0$. The \emph{synchronised concatenation} of $\path_1$ and $\path_2$ is defined only if $\transs_n=\transs'_0$ and is given by t $\path \concatsync \path'= \transs_0\transs_1...\transs_n\transs'_1...\transs'_m$.
\end{defn}

We lift this operations to sets of paths the following way: if $\Path_1$ and $\Path_2$ are two sets of paths, then $\Path_1 \concatsync \Path_2=\{\path_1 \concatsync \path_2 | \path_1 \in \Path_1$ and $\path_2 \in \Path_2\}$.

The authors of  \cite{audinot2017my} introduce the operator of \emph{parallel composition of paths}. However, our definition of attack trees grants us the possibility to use a simpler operator. 

\begin{defn}
Let $\Path_1$, $\Path_2$ be two sets of paths of $\trans$. The \emph{merge} of $\Path_1$ and $\Path_2$ is the set of paths $\Path_1 \andword \Path_2 =\{\path_1 \in \Path_1 |$ there exists $ \path_2 \in \Path_2 $ such that $\path_2\prefix \path_1\} \cup \{\path_2 \in \Path_2 |$ there exists $ \path_1 \in \Path_1 $ such that $\path_1\prefix \path_2\}$
\end{defn}

Unlike the parallel composition of \cite{audinot2017my}, thanks to the transitivity of the prefix relation, the merge operator is associative. 

We can now define our path semantics.

\begin{defn}
\label{Def:semwords}
Let $\att$ be a attack tree over $\Prop$. The path semantics of $\att$ over $\trans$ is the set of paths $\semword{\att}{\trans}$ inductively defined as follow:
\begin{itemize}
\item $\semword{\formula}{\trans} = \{\transs_0\transs_1...\transs_n \in \Path_\trans |  \transs_n \models \formula\}$
\item $\semword{OR(\att_1, ..., \att_n)}{\trans}=\semword{\att_1}{\trans} \cup ... \cup \semword{\att_n}{\trans}$
\item $\semword{SAND(\att_1, ..., \att_n)}{\trans}=\semword{\att_1}{\trans} \concatsync ... \concatsync \semword{\att_n}{\trans}$
\item $\semword{AND(\att_1, ..., \att_n)}{\trans}=\semword{\att_1}{\trans} \andword ... \andword \semword{\att_n}{\trans}$
\end{itemize}
\end{defn}

It is easy to verify that the semantics of Definition \ref{Def:semwords} is equivalent to the one introduced in \cite{audinot2017my} if we restrict to the attack trees whose leaves are of the form $< true,\formula>$. 

Remark that, in our framework, for $\formula_1$ and $\formula_2$ two formulas over $\Prop$, the interpretation of $SAND(\formula_1,\formula_2)$ is that $\formula_1$ must hold at some point and $\formula_2$ must hold at some point afterwards. This requirement does not prevent $\formula_2$ from holding before $\formula_1$.

We also want to point out that our semantics consider that the simultaneity of objectives is always successful: for $\formula_1$ and $\formula_2$ two formulas over $\Prop$, if $\formula_1$, $\formula_2 \in \transv(\transs)$, then $\transs \in \semword{SAND(\formula_1, \formula_2)}{\trans}$ and $\transs \in \semword{AND(\formula_1, \formula_2)}{\trans}$.

\begin{ex}
If we consider the game arena $\arena$ given in \Cref{ga}, we have that $(d_1,d_2) \models \propdone \land \lnot \propseen$, thus the path $(o,m)(d_1,d_2) \in \semword{\propdone \land \lnot \propseen}{\arena}$. This gives us also $(o,m)(d_1,d_2) \in \semword{OR(\propdone \land \lnot \propseen, \propdtwo \land \lnot \propseen )}{\arena}$.
\end{ex}

%% file: sec-multiplayersem.tex

We start this section by formally defining \emph{strategic trees} as well as some handful operators over them. We use a definition of a tree really close to the one made from prefix-closed languages (for example in \cite[p.~15]{comon2008tree}) except that we fix a letter to represent the root.

\begin{defn}
A \emph{strategic tree} (written s-tree for short) over an alphabet $\alphabet$ is a language $\stree$ of the form $\letter \lang$ with $\letter \in \Sigma$ and $\lang$ is a prefix-closed language over $\alphabet$. 
\end{defn}

For an s-tree $\stree=\letter\lang$, $\letter$ is called the \emph{root}. For a word $\word \in \stree$, if there exists no $\word' \in \stree$ such that $\word\prefix\word'$ then we call $\word$ a \emph{leaf}. The set of all leaves of $\stree$ is denoted by $\leaves(\stree)$. For two words $\word, \word' \in \stree$ such that $\word\prefix\word'$, if there exist no $\word'' \in \stree$ such that $\word\prefix\word''\prefix\word'$, then we says that $\word$ is the \emph{parent}  of $\word'$ and $\word'$ is a \emph{child} of $\word$. The set of all children of a word $\word$ in a s-tree $\stree$ is denoted by $\children_\stree(\word)$. The depth of an s-tree is the size of the longest word in it.

\begin{ex}
  \label{Ex:stree}

Figure \ref{Fig:stratstree} shows an s-tree over the alphabet $\Pos$, the set of positions of the game arena of \Cref{ga}.

\end{ex}
As in \Cref{Ex:stree}, for the particular case where alphabet $\alphabet$ is the set of positions on some game arena, we develop several notions on s-trees and show that strategies can be presented as s-trees. 

For the rest of this section we fix a game arena $\arenaa$.

The next lemma asserts that all histories consistent with a strategy and starting from a given position form an s-tree.

\begin{lem}
\label{lem:pref}
Let $\strat$ be a strategy for some player and $\pos \in \Pos$ be a game position. The language $\stree^{\strat}_\pos=Prefixes(\out(\strat)) \cap \Hist(\arena, \pos)$ is an s-tree over alphabet $\Pos$, and is called the \emph{s-tree associated with $\strat$ from position} $\pos$.
\end{lem}

The proof is straightforward from the definition of $\stree^{\strat}_\pos$.

By \Cref{lem:pref}, each branch of $\stree^{\strat}_\pos$ is the succession of all prefixes (in terms of words) of a play consistent with $\strat$. Reciprocally, each play consistent with $\strat$ and starting from position $\pos$ is represented by a branch of $\stree^{\strat}_\pos$. Therefore $\stree^{\strat}_\pos$ fully describes the strategy $\strat$ starting from position $\pos$.

\begin{ex}
\label{Ex:stratstree}
Consider the game arena $\arena$ given in Figure \ref{ga} and the strategy $\strat$ for the thief consisting in waiting one unit of time, then, if the guard is at some door, going to the other door and if the guard is currently in motion, waiting another unit of time before going to the door where the guard will not be. If we call this strategy $\strat$, the strategic tree $\stree^\strat_{(o,d_1)}$ is given in Figure \ref{Fig:stratstree}.

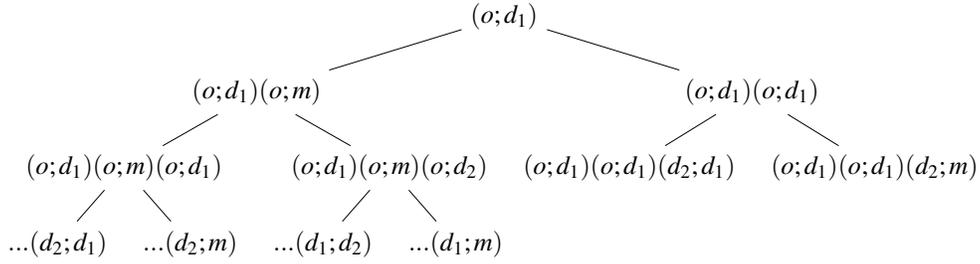
\begin{figure}
  \begin{center}
        \scalebox{0.9}{
\begin{forest}
[$(o ; d_1)$[$(o;d_1)(o;m)$[$(o;d_1)(o;m)(o;d_1)$[$...(d_2;d_1)$][$...(d_2;m)$]][$(o;d_1)(o;m)(o;d_2)$[$...(d_1;d_2)$][$...(d_1;m)$]]][$(o;d_1)(o;d_1)$[$(o;d_1)(o;d_1)(d_2;d_1)$][$(o;d_1)(o;d_1)(d_2;m)$]]]
\end{forest}
}
\end{center}
\caption{strategic tree $\stree^{\strat}_{(o,d_1)}$}
\label{Fig:stratstree}
\end{figure}
\end{ex}

As we put the focus on attack trees, we take the convention that, in the game arena, Player $1$ is called \emph{Attacker} and Player $2$ is called \emph{Defender}. In this setting, Attacker tries to achieve an attack that is described by some attack tree $\att$, while Defender tries to prevent it from happening. In other words, the winning plays for the Attacker are given as $\win_\attacker=\semword{\att}{\arena}$. Our strategy semantics consists of the set of  winning strategies for this game. 

We start by motivating a construction only for a leaf attack tree. The strategy semantics for an attack tree $\phi$ is the set of all strategies that are winning for the reachability game defined by $\phi$.
Remark that for the case of reachability games, once a winning position is reached, the continuation of the play does not matter. Therefore, for reachability games, the s-tree corresponding to a winning strategy can be cut as a finite tree: this cut consists in removing all children of a node describing a history ending in a position where $\formula$ holds. This way of cutting motivates the definition of prefix of s-trees as follows:

\begin{defn}
\label{def:prefix}
Let $\stree$ be an s-tree over $\alphabet^*$. An s-tree $\stree'$ is a \emph{prefix} of $\stree$ if $root(\stree')=root(\stree)$, and $\stree'\subseteq \stree$, and 
for every $\word \in \stree'\setminus \leaves(\stree')$, we have $\children_{\stree'}(\word)=\children_{\stree}(\word)$.
\end{defn}

\begin{ex}
\label{Ex:prefix}
For the s-tree $\stree^\strat_{(o,d_1)}$ of Figure \ref{Fig:stratstree} and the two trees given in Figure $\ref{Fig:prefix}$, we have $\stree_a$ is a prefix of $\stree^\strat_{(o,d_1)}$, but $\stree_b$ is not because $\children_{\stree^\strat_{(o,d_1)}}(o,d_1)=\{(o,d_1)(o,m), (o,d_1)(o,d_1)\}\neq \children_{\stree_b}(o,d_1)\break =\{(o,d_1)(o,d_1)\}$.

\begin{figure}
\centering 
\subfloat{
        \scalebox{0.8}{
\begin{forest}
[$(o ; d_1)$[$(o;d_1)(o;m)$][$(o;d_1)(o;d_1)$[$(o;d_1)(o;d_1)(d_2;d_1)$][$(o;d_1)(o;d_1)(d_2;m)$]]]
\end{forest}}
}
\subfloat{\scalebox{0.8}{
\begin{forest}
[$(o ; d_1)$[$(o;d_1)(o;d_1)$[$(o;d_1)(o;d_1)(d_2;d_1)$][$(o;d_1)(o;d_1)(d_2;m)$]]]
\end{forest}}
}
\caption{two s-trees: $\stree_a$ (left) and $\stree_b$ (right)}
\label{Fig:prefix}
\end{figure}
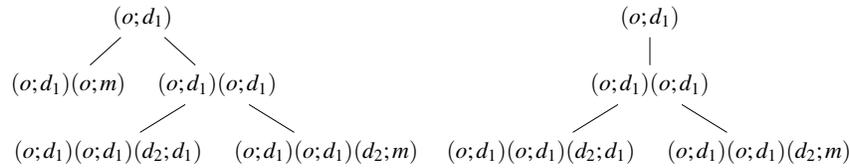
\end{ex}

With this notion of prefix, it is immediate to characterise attack trees that witness a strategy.

\begin{defn}
\label{def:witness}
Consider a leaf attack tree $\formula$, and write $\formula \subseteq Pos$ for the set of positions where $\formula$ holds. Consider $\strat$ a strategy for Attacker in the reachability game $(\arena, \phi)$ and $\stree^{\strat}_\pos$ the associated s-tree  from position $\pos$. A finite s-tree $\cst$ is a \emph{witness} of $\strat$ from position $\pos$ if $\cst$ is a finite prefix of $\stree^{\strat}_\pos$, and $\leaves(\stree) \subseteq \Pos^*\formula$.
\end{defn}

\Cref{def:witness} can be generalised to an arbitrary reachability condition $\reachwin \subseteq \Pos$ as follows: $\stree$ is a witness of $\strat$ from position $\pos$ if $\cst$ is a finite prefix of $\stree^{\strat}_\pos$ and $\hist \in \leaves(\stree)$ implies $last(h)\in \reachwin$.

\begin{ex}
\label{Ex:cst}
In the game arena of \Cref{ga}, if we consider the reachability condition $\reachwin=\{(o,m),\break (d_2,m), (d_2,d_1)\}$, then the attack tree $\stree_a$ of \Cref{Fig:prefix} is a witness for the s-tree $\stree^\strat_{(o,d_1)}$ drawn in \Cref{Fig:stratstree}.
\end{ex}

\Cref{def:witness} leads us to the following intuitive lemma.

\begin{lem}
\label{lem:koning}
Let $\strat$ be a strategy for Attacker and $\reachwin$ be a winning condition. Then $\strat$ is a winning strategy for $(\arena, \reachwin)$ from position $\pos \in \Pos$ if, and only if, there exists a witness $\cst$ of $\strat$ from $\pos$.
\end{lem}

The proof relies on the K\"{o}nig's Lemma.

Thus, for a leaf $\formula$, the strategy semantics is all witnesses that can be constructed from a winning strategy over the reachability game defined by $\formula$. Moreover, an s-tree is in the semantics of a leaf attack tree if it is a prefix of some strategy and if all its leaves are in the path semantics of the attack tree. The former condition guarantees that our s-tree has the shape of a strategy, while the latter guarantees that the strategy is winning. As we will see below, those are the two conditions we use to define the strategy semantics of arbitrary attack trees.


For the first condition, we say that an s-tree $\stree$ is \emph{well-formed} if there exists a strategy $\strat$ and a position $\pos$ such that $\stree$ is a prefix of $\stree^\strat_\pos$. For the second condition, we use the following definition:

\begin{defn}
Let $\att$ be an attack tree. A \emph{$\att$-s-tree} is a finite s-tree $\stree$ over $\Pos$ such that $\leaves(\stree)\subseteq \semword{\att}{\arena}$.
\end{defn}

Since for a leaf attack tree $\formula$, we have $\semword{\formula}{\arena}=\Pos^*\formula $, a witness $\stree$ (\Cref{def:witness}) is a $\formula$-s-tree.
We now have all the material to define the strategy semantics of an attack tree.

\begin{defn}
Let $\att$ be an attack tree. The \emph{strategy semantics} associated with $\att$, written $\sem{\att}{\arena}$ is the set of all well-formed $\att$-s-trees.
\end{defn}

In particular, $\sem{\formula}{\arena}$ is the set of all witnesses in the reachability game $(\arena,\formula)$.

We can see that the idea is far from the one of attack-defence trees in \cite{kordy2010foundations}. In attack-defence trees, the countermeasure is a structure similar to an attack tree whose semantics describes paths that prevent an attack from succeeding, and by no means a strategy of the attacker's opponent in the arena.

Now that we defined our semantics, we might want to know if it can be obtained in a compositional manner
? Namely, if the semantics of a compound tree can be defined in terms of the semantics of its subtrees: More formally.. can we define $\sem{OP(\att_1, ..., \att_n)}{\arena}$ on the basis of $\sem{\att_1}{\arena}, ...,\break \sem{\att_n}{\arena}$? Sadly, the answer is no: 

\begin{ex}
\label{Ex:OR}
Consider the game arena defined in Figure \ref{ga}. Obviously, our attacker here will be the thief while the guard will do the defender role. We also consider a new proposition: $Start$ which only holds at position $\{o,m\}$. We have that the semantics of $SAND(start, \propdone\land \lnot \propseen)$ is empty. Indeed, the guard can choose to only keep door $1$ and thus, the thief will not be able to attain $\propdone$ while remaining unseen. Similarly, $\sem{SAND(start, \propdtwo\land \lnot \propseen)}{\arena}$ is empty. However, the strategy consisting on waiting one unit of time then going through the door not controlled by the guard is a winning strategy, it is easy to construct a witness for that strategy that attains the objective of $OR(SAND(start, \propdone\land \lnot \propseen),SAND(start,\propdtwo\land \lnot \propseen))$ and thus is in its strategy semantics.
\end{ex}

The previous example showcases an empty semantics for $\att_1$ and $\att_2$ but a non-empty one for $OR(\break \att_1, \att_2)$. This is because, for $\formula_1$ and $\formula_2$ two propositional formulas over $\Prop$, there are more strategies to achieve $\formula_1\lor\formula_2$ than strategies only achieving $\formula_1$ or only achieving $\formula_2$. We can for example consider a strategy that, depending on the move of the opponent, chooses whether it prefers to attain $\formula_1$ or to attain $\formula_2$.

Remark that, using the "merge" operator of \cite{paul2015automata} provides us a compositional semantics for attack trees with $SAND$-only operators. However, we have already argues that the $OR$ operator have some problems just as the $AND$ operator for more elaborate examples. Still, it is possible to tune the semantics so that it becomes compositional for the AND operator, at the price of loosing clarity, but more regrettably without solving the hopeless case of the OR operator.


%% file: sec-decisionproblems.tex
In this section, we discuss two common decision problems over semantics of attack trees and determine their complexities with respect to the path semantics and the strategy semantics. The first problem we consider is the Non-Emptiness problem. This problem consists of, given an attack tree and a game arena, deciding whether its semantics is not empty:

\begin{defn}
The \emph{Non-Emptiness} problem is the following decision problem for a fixed semantics $\semgeneral{\cdot}{\arena}$ of attack trees:
\newline \textbf{Input:} $\arena$, a game arena, $\att$, an attack tree.
\newline \textbf{Output:} $Yes$ if $\semgeneral{\att}{\arena}\neq \emptyset$, $No$ otherwise.
\end{defn}

The Non-Emptiness problem for the path semantics is denoted
by PNE while the Non-Emptiness problem for the strategy semantics
is denoted SNE. A positive instance of PNE tells us that Attacker has a
favourable scenario to attack. A positive instance of SNE tells us that Attacker has a strategy (it
is possible for him to attack successfully the system independently of
the defender/environment comportment).

We now turn to the Membership problem.

\begin{defn}
The \emph{Membership problem} is the following decision problem for a fixed attack tree semantics $\semgeneral{\cdot}{\arena}$ of  of type $X$: 
\newline \textbf{Input:} $\arena$, a game arena, $\att$, an attack tree and $x \in X$.
\newline \textbf{Output:} Yes if $x \in \semgeneral{\att}{\arena}$, No otherwise.
\end{defn}

The Membership problem for the path semantics is denoted by PM while the Membership problem for the strategy semantics is denoted SM. PM consists of determining whether a path is an attack or not. It can be really useful if we have an attack tree describing an attack goal over an information system and a log file of that system. Determining if the system has been attacked is equivalent to determining whether the path described by the log file is in the path semantics of the attack tree or not. The idea behind SM is different: it is useful to determine whether a strategy is winning or not for a given attack objective. We start to analyse the complexity of PM and take advantage of it for the proofs of the other results. We then consider SNE. After that, PNE is easily determined as a particular case of SNE and we finish by SM whose proof uses similar and simpler constructions than the one for SNE.


If we use attack trees with preconditions, the problem PM is \NP-hard; this comes from the fact that the \emph{packed interval covering problem}, which can be easily captured by the parallel composition (see \cite{pinchinat2020library}), is \NP-complete (see \cite{saffidine2019packed}). However, PM becomes simpler if we discard preconditions:

\begin{thm}
 \label{Thm:polytimemem}
PM is in \PTIME.
\end{thm}

For a polynomial algorithm, we use the fact that a word is in
the semantics of an attack tree, then adding an arbitrary prefix to it
keeps it in the semantics. As a consequence, we do not need to recompute
which sub-goals of the attack tree are satisfied whenever we add a position in front of a path. Thus, the shape of the problem is well-suited for a backward
induction over the input path. Moreover, determining if a given input
path satisfies an attack tree knowing whether it satisfies the sub-trees can be done in linear time over the size of the
attack tree.


We now turn to the complexity of SNE.

\begin{thm}
\label{thm:sne}
SNE is \PSPACE-complete.
\end{thm}

For the membership, we construct an alternating algorithm (see \cite{chandra1976alternation})
solving the problem that can be executed in polynomial time. This
algorithm consists of synthesizing a history over the game arena and
then verifying that this history is an attack (by \Cref{Thm:polytimemem}, this verification is doable in polynomial time). To construct this history, we finitely iterate first to make a non-deterministic existential guess for
the action of Attacker and then a non-deterministic universal
guess for the action of Defender. We then
show that the resulting history is in the path semantics of the input attack tree $\att$  if,
and only if, the strategy semantics of $\att$ is not empty. We guarantee a polynomial time execution, namely that the resulting
history need not be too long with the following lemma.

\begin{lem}
\label{Lem:stratlength}
Let $\arenaa$ be a game arena and $\att$ be an attack tree with $n$ leaves. If $\sem{\att}{\arena} \neq \emptyset$, then there exists $\stree \in \sem{\att}{\arena}$ of depth $d \leq |\Pos| \times n$.
\end{lem}

The basic idea behind to prove \Cref{Lem:stratlength} is that, memoryless strategies suffice in reachability games (see \cite{de2007concurrent}).

We design Algorithm \ref{Algo:emptymulti} to solve SNE whose idea is explained above and show that it belongs to \PSPACE.

\begin{algorithm}
\caption{$SNE(\arena, \att$)}
\label{Algo:emptymulti}
\textbf{Input:} $\arena$ a game arena and $\att$ an attack tree with $n$ leaves \newline
\textbf{Output:} $True$ if $\sem{\att}{\arena}\neq \emptyset$, $False$ otherwise.

\begin{algorithmic}[1]

\STATE $\hist \leftarrow$ empty list
\STATE $\pos \leftarrow$ [$\exists$]guess position in $\Pos$
\STATE $\hist.append(\pos)$
\WHILE {$size(\hist)< |Pos| \times n$}
\STATE [$\exists$]guess break or not
\STATE $\act_1 \leftarrow$ [$\exists$]guess action in $\Act_\attacker$
\STATE $\act_2 \leftarrow$ [$\forall$]guess action in $\Act_\defender$
\STATE $\hist.append(\arenaf(last(\hist),(\act_1, \act_2)))$
\ENDWHILE
\RETURN $\hist \in \semword{\att}{\arena}$
 
\end{algorithmic}
\end{algorithm}

\begin{lem}
Algorithm \ref{Algo:emptymulti} is an alternating polynomial-time algorithm and solves SNE.
\end{lem} 

\begin{proof}

We start by showing the complexity of the algorithm, then we show its correctness.

From the loop at Line 4, it is executed polynomially many times in the size of the input attack tree and of the game arena. We also know (\Cref{Thm:polytimemem}) that the condition $\hist \in \semword{\att}{\arena}$ at Line 10 can be evaluated in polynomial, therefore, Algorithm \ref{Algo:emptymulti} is polynomial-time alternating.

Assume Algorithm \ref{Algo:emptymulti} returns $True$, then, for each choice made by universal guess, there exists a choice made by existential guess guaranteeing that the obtained history is in $\semword{\att}{\arena}$. As a consequence, the choices made by the existential guesses reflect a strategy in the game arena that satisfies $\att$ so, $\sem{\att}{\arena}\neq \emptyset$. Conversely, if $\sem{\att}{\arena}\neq \emptyset$, then there exists (by Lemma \ref{Lem:stratlength}) an s-tree $\stree$ of depth $\leq |Pos| \times n$. Thus the existential guesses can simply follow the strategy given by $\stree$ and then choose to go out from the main loop by the "break" command at Line 5 of Algorithm \ref{Algo:emptymulti} whenever the sequence of choices (existential and universal) in the execution is reflected by a full branch of the s-tree $\stree$.

\end{proof}

For the \PSPACE-hardness of SNE, our construction is inspired by the one in \cite{audinot2017my}: the authors reduce (in polynomial time) the SAT problem to the PNE problem with attack trees (using preconditions). In fact, even if in that paper, authors use attack trees with preconditions, we can adapt it without preconditions. We can even cast the approach to QBF that we first recall:

\begin{defn}
The \emph{quantified Boolean formula} (QBF) is the following decision problem:
\newline \textbf{Input:} a formula of the form $Q_1 x_1, ..., Q_n x_n \psi(x_1, ..., x_n)$ with $Q_i \in \{\exists, \forall\}$ and $\psi$ a Boolean formula in conjunctive normal form over propositions $x_1, ..., x_n$.
\newline \textbf{Output}: $Yes$ if the input formula is true, $No$ otherwise.
\end{defn}

\begin{lem}
\label{Lem:qbf}
The QBF problem can be reduced to SNE in polynomial time.
\end{lem}

It is easy to understand the reduction principle on an example.

\begin{ex}
\label{Ex:qbf}
Consider the formula $\psi=\exists x_1 \forall x_2 \exists x_3, x_1 \land
(x_2 \lor x_3) \land (\lnot x_2 \lor x_3)$. Let $C_1=x_1$, $C_2=(x_2 \lor x_3)$ and $C_3=(\lnot x_2 \lor x_3)$ be the three clauses in $\psi$. The game arena $\arena$
associated with this formula is drawn in Figure \ref{Fig:qbf}: for each position $v_i$ (resp. $\lnot v_i$), the proposition $p_i$ holds if $v_i \in C_i$ (resp. $\lnot v_i \in C_i$). Remark that this game arena is a special case of game arena called turn-based game arena: only one player makes an action in each position, we say that a position belongs to the player who can play on it. We decide classically which position belongs to each player based on quantifiers of $\psi$ (see the proof of \Cref{Lem:qbf} for further explanations). We represent Attacker positions with a circle and Defender positions with a square (position $\pos_3$ and position $\lnot pos_3$ have only one successor position, therefore, it does not matter which player makes the move; by convention, we say they belong to the attacker). Then, $\psi$ holds if, and only if, $\sem{SAND(start, AND(p_1,p_2,p_3))}{\arena}\neq \emptyset$. 
\end{ex}

\begin{figure}
\centering
\subfloat{
\begin{tikzpicture}[scale=0.6,node distance = 0.1cm and 0.5cm]

   \node[state] (q0)   {$Start$}; 
   \node[draw, regular polygon, regular polygon sides=4, minimum size=1.6cm] (a) [above right=of q0] {$\pos_1$};
   \node[draw, regular polygon, regular polygon sides=4, minimum size=1.5cm] (la) [below right=of q0] {$\lnot \pos_1$};
   \node[state] (b) [ right=of a] {$\pos_2$};
   \node[state] (lb) [ right=of la] {$\lnot \pos_2$};
   \node[state] (c) [ right=of b] {$\pos_3$};
   \node[state] (lc) [ right=of lb] {$\lnot \pos_3$};
   
   \node (bu) at (-2,0) {$\{start\}$};
   \node (blu) at (2.5,3.2) {$\{p_1\}$};
   \node (blu') at (2.5,-3.2) {$\emptyset$};
   \node (r) at (5,3.2) {$\{p_2\}$};
   \node (s) at (5,-3.2) {$\{p_3\}$};
   \node (t) at (7.5, 3.2) {$\{p_2,p_3\} $};
   \node (u) at (7.5, -3.2) {$ \emptyset$};

   \path[->] 
    (q0) edge node {} (a)
          edge node {} (la)
    (a) edge node {} (b)
    edge node {} (lb)
    (la) edge node {} (b)
    edge node {} (lb)
    (b) edge node {} (c)
    edge node {} (lc)
    (lb) edge node {} (c)
    edge node {} (lc)
    (c) edge [loop right] node {} ()
    (lc) edge [loop right] node {} ()
    ;
\end{tikzpicture}
}
\subfloat{
\begin{tikzpicture}[scale=1.4]
\node [draw, ellipse] (p) at (0,0) {};
\node [draw, ellipse] (q) at (-1,-1.2) {$start$};
\node [draw, ellipse] (r) at (1,-1.2) {};
\node [draw, ellipse] (s) at (0,-2.4) {$p_1$};
\node [draw, ellipse] (t) at (1,-2.4) {$p_2$};
\node [draw, ellipse] (u) at (2,-2.4) {$p_3$};

\draw[->] (-0.5,-0.6) to[bend right] (0.5,-0.6);
\draw[-] (0.5,-1.8) to[bend right] (1.5,-1.8);
\draw (q) -- (p) -- (r) -- (s);
\draw (t)--(r)--(u);

\end{tikzpicture}
}

\caption{Game arena and attack tree associated to the formula given in Example \ref{Ex:qbf}}
\label{Fig:qbf}
\end{figure}
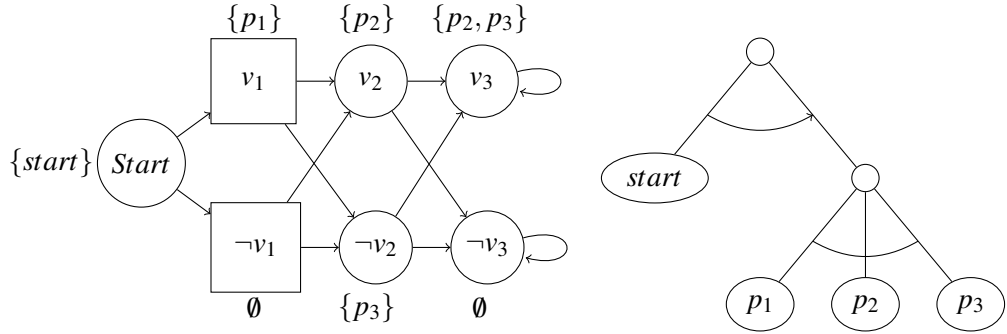

We now start the proof of \cref{Lem:qbf}:

\begin{proof}
Let $Q_1 x_1, ..., Q_n x_n \psi(x_1, ..., x_n)$ with $Q_i \in \{\exists, \forall\}$ and with $\psi$ a Boolean formula over variables $x_1, ..., x_n$ be an instance of the QBF problem. Since $\psi$ is in conjunctive normal form, we can write it as $\psi=\psi_1 \land ... \land \psi_k$ with $\psi_i$ denoting disjunctive clauses containing literals of the form $x_j$ or $\lnot x_j$ with $x_i \in \{x_1, ..., x_n\}$.

We consider the set of propositions $\Prop=\{Start, p_1, ..., p_k\}$
with the following game arena:\break $\arenaa$, where $\Pos=\{Start\} \cup
\{\pos_i | 1 \leq i \leq n\} \cup \{\lnot \pos_i | 1 \leq i \leq
n\}$, $\Act_\attacker = \Act_\defender = \{True,\break False\}$. If $Q_0=\exists$, then position $Start$ is an Attacker position, otherwise, it's a defender position. Moreover, playing action $True$ at position $start$ leads to position $\pos_1$ while playing $False$ leads to position $\lnot pos_1$. Similarly, for each $2 \leq i \leq n$, if $Q_i= \exists$ then $\pos_{i-1}$ and $\lnot\pos_{i-1}$ are Attacker positions, otherwise, they are Defender positions. Furthermore, playing $True$ at position $\pos_{i-1}$ or $\lnot\pos_{i-1}$ leads to position $\pos_i$ while playing $False$ leads to $\lnot\pos_{i}$. Positions $\pos_n$ and $\lnot \pos_n$ are Attacker positions, moreover, the transitions over those two positions are self loops. 

We define $\arenav(Start)=\{Start\}$ and for each $i \leq i \leq n$,
$\arenav(\pos_i)=\{p_j | x_i \in \psi_j\}$ and $\arenav(\lnot
\pos_i)=\{p_j | \lnot x_i \in \psi_j\}$. From this definition, if we
consider that the attacker tries to satisfy the input QBF formula
and the defender tries to prevent it, we have a classic game. We then
only need to show that the objective of the attacker can be well
described using an attack tree, which is the case by considering
$\att=SAND(Start, AND(p_1, ..., p_n))$. Indeed, if there exists a
strategy to satisfy the input QBF formula, then this strategy
satisfies $\psi_1, ..., \psi_k$ and thus, can be executed in the
constructed game arena to achieve $AND(p_1, ..., p_n)$ while starting
at position $Start$, therefore, that strategy is in
$\att$. Conversely, if $\sem{\att}{\arena}\neq\emptyset$, then one of such strategies assures that we satisfy the
input QBF instance.
\end{proof}

By \Cref{Lem:qbf} , SNE is \PSPACE-hard, which achieves the proof of \Cref{thm:sne}.


We now turn to PNE.

\begin{thm}
PNE is \NP-complete.
\end{thm}

For the \NP-membership, since our problem is a particular case of the problem discussed in \cite{audinot2017my}, it is at least as easy. For the \NP-hardness we reduce SAT: if we apply the same construction as in the proof of Lemma \ref{Lem:qbf}, since we cannot leave any choice for the defender in a transition system and the path semantics is defined over a transition system and not a game arena, we can reduce formulas of QBF only using $\exists$ operators. In other words, we can reduce SAT. In fact, by doing so, we are doing the exact construction of the proof in \cite{audinot2017my}. Moreover the attack tree with preconditions $AND(<start, \formula_1>, ..., <start, \formula_n>)$ used in that paper is completely equivalent to $SAND(start, AND(\formula_1, ..., \formula_n))$ in our formalism. Thus the proof in \cite{audinot2017my} can be well adapted for our problem.


Lastly, we study SM.

\begin{thm}
\label{thm:SM}
SM is \coNP-complete.
\end{thm}

For the membership, we can use the same idea as for the membership of the SNE except that, now, we already know the strategy of the attacker, we thus do not need to use any existential guess for the action of Attacker. In other words, it is equivalent to simply considering Defender choosing a branch of the attack tree and  then verifying if it forms an attack or not. Therefore, we use a variant of Algorithm \ref{Algo:emptymulti} without existential choices, this gives us a \coNP algorithm.

For the hardness, we still use the idea of the construction behind the SNE, but now, we consider that only the actions of Defender matter in the progress of the game arena. This way, we can reduce the UNSAT problem, known to be \coNP-complete, to SM. The UNSAT problem is nothing less than the sub-problem  of the QBF problem where an instance of the problem only uses "$\forall$" quantifiers.

This concludes the discussion over decision problems; our results are summarised in \Cref{table:results}.

%% file: sec-future.tex
In this paper, we proposed a strategy semantics for attack trees, useful to tackle some practical questions (SNE and SM) not expressible with standard semantics provided by the literature. The price to pay is to renounce a compositional semantics of attack trees. One way to regain it might be to consider a strategy semantics based on a tree automata: we associate with each attack tree a tree automaton recognising its strategy semantics. This is currently work. Moreover, being able to consider automata recognising the strategy semantics allows us to model attack scenarios with constraints, for example, considering that the attacker cannot perform a given action more than a certain amount of time.

Moreover, we are currently exploring the possibility to expand the path and the strategy semantics to attack-defense trees. The main idea is to consider a counter operator in attack trees. This generalisation could lead to a better understanding of the differences between the strategy semantics and the attack-defence tree formalism.